\documentclass[11pt,english,aps,pra,onecolumn,tightenlines,superscriptaddress,notitlepage,floatfix,fleqn]{revtex4-1}
\renewcommand{\bibliography}[1]{}
\pdfoutput=1
\usepackage[utf8]{inputenc}
\usepackage[T1]{fontenc}
\usepackage{amsmath}
\usepackage{amssymb}
\usepackage{amsfonts}
\usepackage{bm}
\usepackage{bbm}
\usepackage{times}
\usepackage{enumitem}
\usepackage{mathrsfs}

\usepackage[usenames,dvipsnames]{color}
\usepackage{tikz}
\usetikzlibrary{automata,arrows,positioning,calc}
\definecolor{dblue}{rgb}{0,0.1,.6}

\usepackage[colorlinks=true,citecolor=dblue,linkcolor=dblue,urlcolor=dblue]{hyperref}
\usepackage[all]{hypcap}

\newcommand{\id}{\mathbbm{1}}
\newcommand{\bra}{\langle}
\newcommand{\ket}{\rangle}
\newcommand{\Tr}{\operatorname{Tr}}
\newcommand{\Span}{\operatorname{span}}
\newcommand{\supp}{\operatorname{supp}}
\renewcommand{\vec}[1]{{\boldsymbol{#1}}}
\newcommand{\ud}{\mathrm{d}}
\newcommand{\CC}{\mathbb{C}}

\newcommand{\mc}[1]{\mathcal{#1}}
\newcommand{\pdag}{{\phantom{\dag}}}

\newcommand{\A}{\mathscr{L}}
\renewcommand{\H}{\mc{H}}
\renewcommand{\L}{\mc{L}}
\renewcommand{\P}{\mc{P}}
\newcommand{\E}{\mc{E}}
\newcommand{\V}{\mc{V}}
\newcommand{\W}{\mc{W}}
\newcommand{\veps}{\varepsilon}

\newcommand{\hA}{\hat{A}}
\newcommand{\hB}{\hat{B}}
\newcommand{\hC}{\hat{C}}
\newcommand{\hc}{\hat{c}}
\newcommand{\hH}{\hat{H}}
\newcommand{\hK}{\hat{K}}
\newcommand{\hL}{\hat{L}}
\newcommand{\hM}{\hat{M}}
\newcommand{\hO}{\hat{O}}
\newcommand{\hP}{\hat{P}}
\newcommand{\hU}{\hat{U}}
\newcommand{\hsigma}{\hat{\sigma}}
\newcommand{\dm}{{\hat{\rho}}}
\newcommand{\mri}{\mathrm{i}\mkern1mu}
\newcommand{\up}{{\uparrow}}
\newcommand{\down}{{\downarrow}}
\newcommand{\vk}{\vec{k}}

\renewcommand{\Bmatrix}[1]{\begin{bmatrix}#1\end{bmatrix}}
\newcommand{\Pmatrix}[1]{\begin{pmatrix}#1\end{pmatrix}}

\usepackage{amsthm}
\newtheorem{theorem}{Theorem}[section]
\newtheorem{proposition}[theorem]{Proposition}
\newtheorem{corollary}[theorem]{Corollary}

\newtheorem{example}{Example}[section]
\newtheorem{definition}{Definition}[section]

\newtheorem*{remark}{Remark}

\makeatletter
\def\@hangfrom@section#1#2#3{\@hangfrom{#1#2}#3}
\def\@hangfroms@section#1#2{#1#2}
\makeatother

\newcommand{\duke}  {Department of Physics, Duke University, Durham, North Carolina 27708, USA}
\newcommand{\dqc}   {Duke Quantum Center, Duke University, Durham, North Carolina 27701, USA}

\begin{document}

\title{Criteria for Davies Irreducibility of Markovian Quantum Dynamics}
\author{Yikang Zhang}
\affiliation{\duke}
\author{Thomas Barthel}
\affiliation{\duke}
\affiliation{\dqc}

\begin{abstract}
The dynamics of Markovian open quantum systems are described by Lindblad master equations, generating a quantum dynamical semigroup. An important concept for such systems is (Davies) irreducibility, i.e., the question whether there exist non-trivial invariant subspaces. Steady states of irreducible systems are unique and faithful, i.e., they have full rank. In the 1970s, Frigerio showed that a system is irreducible if the Lindblad operators span a self-adjoint set with trivial commutant. We discuss a more general and powerful algebraic criterion, showing that a system is irreducible if and only if the multiplicative algebra generated by the Lindblad operators $L_a$ and the operator $K=iH+\sum_a L^\dagger_aL_a/2$, involving the Hamiltonian $H$, is the entire operator space. Examples for two-level systems, show that a change of Hamiltonian terms as well as the addition or removal of dissipators can render a reducible system irreducible and vice versa. Examples for many-body systems show that a large class of spin chains can be rendered irreducible by dissipators on just one or two sites.
Additionally, we discuss the decisive differences between (Davies) reducibility and Evans reducibility for quantum channels and dynamical semigroups which has lead to some confusion in the recent physics literature, especially, in the context of boundary-driven systems. We give a criterion for quantum reducibility in terms of associated classical Markov processes and, lastly, discuss the relation of the main result to the stabilization of pure states and argue that systems with local Lindblad operators cannot stabilize pure Fermi-sea states.
\end{abstract}

\date{September 27, 2023}

\maketitle

\vspace{1em}
\renewcommand{\baselinestretch}{0.7}\normalsize
\tableofcontents
\renewcommand{\baselinestretch}{1}\normalsize

\newpage
\section{Introduction}
Recent experimental progress on highly controllable quantum systems has motivated a wave of research on open quantum systems, where driving and dissipation can lead to undesirable decoherence but may also stabilize interesting states and lead to new physics. For many open systems there is a clear separation between the system and environment time scales such that the dynamics becomes \emph{Markovian}. A typical case are systems that are weakly coupled to the environment and described in the framework of the Born-Markov and secular approximations \cite{Alicki2007,Breuer2007}. The evolution of Markovian open systems is characterized by a \emph{Lindblad master equation} \cite{Lindblad1976-48, Gorini1976-17}
\begin{align}\label{eq:QME}
	\partial_t \dm=\L(\dm)=-\mri[\hH,\dm]+\sum_{\alpha}\left(\hL_\alpha\dm \hL_\alpha^\dag-\frac{1}{2}\{\hL_\alpha^\dag \hL_\alpha,\dm\}\right)
\end{align}
where $\dm$ is the system density operator, $\hH$ is the Hamiltonian, and the Lindblad operators $\hL_\alpha$ capture the environment coupling, which can lead to decoherence and dissipation. If the master equation is derived from a Hamiltonian description of a system weakly coupled to environmental degrees of freedom, the Hamiltonian $\hH$ in Eq.~\eqref{eq:QME} comprises both the system Hamiltonian and the Lamb-shift Hamiltonian \cite{Alicki2007,Breuer2007}.
We consider time-independent Liouvillians $\L$, where steady states $\dm_\text{ss}$ obey $\L(\dm_\text{ss})=0$.

Fundamental questions concern conservation laws and the existence of unique steady states, which have important consequences concerning the preparation of quantum states and driven-dissipative phase transitions. Criteria for the uniqueness of the steady state have been studied extensively in the past 50 years with seminal works like \cite{Davies1970-19,Evans1977-54,Evans1978-2-17,Frigerio1977-2,Frigerio1978-63,Spohn1980-52} and more recent contributions such as \cite{Fagnola1998-1,Schrader2001-30,Fagnola2002-43,Umanita2005-134,Schirmer2010-81,Carbone2016-77,Nigro2019-4}.

The main topic of this work is a necessary and sufficient algebraic criterion for Davies irreducibility based on the Hamiltonian and Lindblad operators, which makes it possible to efficiently assess irreducibility for large classes of many-body systems.
Sections~\ref{sec:motivation} and \ref{sec:distinction} explain and illustrate the decisive differences between the concepts of Evans irreducibility and (Davies) irreducibility, which caused some confusion in the physics literature. Only Davies irreducibility guarantees the uniqueness and faithfulness of steady states \footnote{According to Brouwer's fixed point theorem, every finite-dimensional Markovian system has at least one steady state. Infinite-dimensional systems need not have steady states and can be unstable. See, for example, Refs.~\cite{Davies1970-19,Barthel2021_12}.}.
In Section~\ref{sec:classical}, we relate quantum channels to classical Markov processes and point out a relation between Davies irreducibility for the quantum system and the classical irreducibility of the associated Markov processes.
In Section~\ref{sec:algebraic}, we leverage approaches to the characterization of invariant subspaces to formulate algebraic criteria for Davies irreducibility. According to the main result, a system is irreducible if and only if the Lindblad operators $\{\hL_\alpha\}$ and the operator $\hK=\mri\hH+\frac{1}{2}\sum_\alpha \hL^\dagger_\alpha\hL_\alpha$ generate under multiplication and linear combination the entire operator algebra of the system \footnote{Importantly, we consider the associative $\CC$-algebra $\W$ generated by $\{\hL_\alpha,\hK\}$ through multiplication and linear combination. Hence, $\W$ does not necessarily contain $\hL_\alpha^\dag$ or the Hamiltonian $\hH$, and it is not necessarily a $C^*$-algebra.}. From this, one can deduce less general criteria employing the commutant of a subset of operators from the algebra that $\hL_\alpha$ and $\hK$ generate, including Frigerio's second theorem \cite{Frigerio1978-63} (Corollary~\ref{coro:Frigerio2} below).
Examples for two-level systems are presented in Section~\ref{sec:examples-qubit}, demonstrating that a change of Hamiltonian terms as well as the addition or removal of dissipators can render a reducible system irreducible and vice versa.
Some applications to many-body systems are discussed in Sections~\ref{sec:examples-mb1} and \ref{sec:examples-mb2} showing, for example, that one-dimensional many-body models can be rendered irreducible by having suitable dissipators on just one or two sites. This includes a corrected argument for the uniqueness of the steady state in XXZ spin-1/2 chains with boundary driving $\hL_1\propto \hsigma^+_1$ and $\hL_2\propto \hsigma^-_N$ at the two ends as considered in Refs.~\cite{Prosen2012-86,Prosen2015-48} and subsequent works.
In Section~\ref{sec:stabilization}, we show that a criterion for the driven-dissipative stabilization of pure states follows from the main result, and we argue that pure Fermi-sea states cannot be stabilized when all dissipators are spatially local.

\section{Motivation: Evans theorem and irreducibility}\label{sec:motivation}
In 1977, Evans published a paper discussing the irreducibility of quantum dynamical semigroups \cite{Evans1977-54} generated by Liouville super-operators (a.k.a.\ Liouvillians) $\L$ as characterized by Eq.~\eqref{eq:QME}. 
For the following, let  $\H$ denote the Hilbert space of the system, which may be finite or infinite-dimensional. $\A(\H)$ denotes the set of all linear operators on $\H$, and recall that the \emph{commutant} $\mc{X}'$ of a set of operators $\mc{X}\subseteq\A(\H)$ is defined as the set of operators $\mc{Y}\subseteq\A(\H)$ such that $[\hat{X},\hat{Y}]=0$ $\forall \hat{X}\in\mc{X}$ if and only if $\hat{Y}\in\mc{Y}$. Furthermore, recall that an observable $\hO$ is \emph{conserved} if $\L^\dag(\hO)=0$ such that, in the Heisenberg picture with $\Tr[\hB\L(\hA)]\equiv \Tr[\L^\dag(\hB)\hA]$ and $\hO(t)=e^{\L^\dag t}\hO$, we have
\begin{equation}
	\partial_t\hO(t)=\L^\dag\big(\hO(t)\big)=0.
\end{equation}

Evans established a theorem connecting the algebraic properties of $\hL_\alpha$, $\hH$ and the conservation of projection operators (theorem~4.1 in Ref.~\cite{Evans1977-54}):
\begin{theorem} [Evans theorem]
  \label{theorem:Evans}
  For a Markovian system \eqref{eq:QME}, a projection operator $\hP=\hP^\dag=\hP^2$ is conserved if and only if $\hP$ is an element of the commutant $\{\hL_\alpha, \hL_\alpha^\dag,\hH\}'$.
\end{theorem}
\begin{proof} 
The `if' part of the theorem is obvious. If $\hP$ commutes with all $\hL_\alpha$, $\hL_\alpha^\dag$, and $\hH$, then
\begin{align}\label{eq:LME-Heisenberg}
	\L^\dag(\hP)\stackrel{\eqref{eq:QME}}{=}
	\mri[\hH,\hP]+\sum_{\alpha}\left(\hL^\dag_\alpha\hP \hL_\alpha-\frac{1}{2}\{\hL_\alpha^\dag \hL_\alpha,\hP\}\right)=\sum_{\alpha}\left(\hP\hL^\dag_\alpha \hL_\alpha-\hP\hL^\dag_\alpha \hL_\alpha\right)=0.
\end{align}
To prove the `only if' part of the theorem, we assume $\L^\dag(\hP)=0$ and use $\hP=\hP^\dag=\hP^2$ to find that
\begin{align}
	0=\L^\dag(\hP^\dag\hP)-\L^\dag(\hP^\dag)\hP-\hP^\dag\L^\dag(\hP)=\sum_\alpha (\hP\hL_\alpha-\hL_\alpha\hP)^\dag(\hP\hL_\alpha-\hL_\alpha\hP).
\end{align}
As the right-hand side is a sum of positive semidefinite operators, This leads to $[\hP,\hL_\alpha]=0$ $\forall \alpha$. Since $\hP$ is Hermitian, $[\hP,\hL^\dag_\alpha]=0$ $\forall \alpha$ follows immediately. Inserting this into Eq.~\eqref{eq:LME-Heisenberg}, we finally have $0=\L^\dag(\hP)=\mri[\hH,\hP]$. Thus, we conclude that $\hP\in\{\hL_\alpha,\hL_\alpha^\dag,\hH\}'$ if $\L^\dag(\hP)=0$.
\end{proof}

An alternative proof of this result can be found in Ref.~\cite{Baumgartner2008-41} (lemma~7).
Now, Evans defined the absence of non-trivial conserved projection operators $\hP\neq 0, \id$ as irreducibility of the Liouvillian. Previous works by Davies \cite{Davies1970-19} established a connection between irreducibility and the occurrence of a unique full-rank steady state; cf.\ Theorem~\ref{theorem:faithful} below. In the recent physics literature, these results are sometimes (mis)interpreted as implying that the steady state is unique if and only if the set of operators $\{\hL_\alpha,\hL_\alpha^\dag,\hH\}$ generates the entire operator algebra $\A(\H)$ for the Hilbert space $\H$ of the system. See, e.g., Refs.~\cite{Popkov2012-12,Popkov2013-15,Prosen2015-48,Manzano2018-67,Nigro2019-4}.

However, this is generally not true, and the following Liouvillian for a two-spin system provides a counter-example. Its many-body version has been used in Ref.~\cite{Barthel2020_12} to illustrate a no-go result for driven-dissipative phase transitions, and it has also been studied in a different context in Ref.~\cite{Lenarcic2020-125}.
\begin{example}[Two-site ferromagnet]
\label{example1}
Consider two spins-$1/2$ with no Hamiltonian and the Lindblad operators
\begin{align}
	\hL_1=\hsigma_1^z,\quad
	\hL_2=\hsigma_2^z,\quad
	\hL_{12}^+=\hP_1^\up\hsigma_2^+,\quad
	\hL_{12}^-=\hP_1^\down\hsigma_2^-,\quad
	\hL_{21}^+=\hP_2^\up\hsigma_1^+,\quad
	\hL_{21}^-=\hP_2^\down\hsigma_1^-,
\end{align}
where $\{\hsigma^a\}$ are Pauli matrices, $\hP^\up\equiv|\up\ket\bra\up|$, $\hP^\down\equiv|\down\ket\bra\down|$, $\hsigma^-\equiv|\down\ket\bra\up|$, and $\hsigma^+\equiv|\up\ket\bra\down|$.
The two pure states $|\up\up\ket\bra\up\up|$ and $|\down\down\ket\bra\down\down|$ are both steady states of the Liouvillian. Further, we observe that
\begin{align}
	\hL_{12}^++\hL_{12}^{-\dag}=\hsigma_2^+,\quad 
	\hL_{12}^-+\hL_{12}^{+\dag}=\hsigma_2^-,\quad 
	\hL_{21}^++\hL_{21}^{-\dag}=\hsigma_1^+,\quad 
	\hL_{21}^-+\hL_{21}^{+\dag}=\hsigma_1^-.
\end{align}
Combining this with the fact that the commutation operation is a linear operation, one has
\begin{align}
	\{\hsigma_1^-,\hsigma_1^+,\hsigma_1^z,\hsigma_2^-\,\hsigma_2^+,\hsigma_2^z\}'\subseteq \{\hL_\alpha^\pdag, \hL_\alpha^\dag\}'.
\end{align}
Finally, since the associative $\CC$-algebra generated by the elements of the left set is the full operator space $\A(\H)$, Schur's lemma gives
\begin{align}
	\{\hL_\alpha, \hL_\alpha^\dag\}'=\{z\id\,|\, z\in\CC\}.
\end{align}
In conclusion, we have constructed a Liouvillian with multiple steady states but no nontrivial conserved projection operators.
\end{example}

We will describe further counter-examples in Sec.~\ref{sec:examples-qubit}. As discussed in the following, the apparent contradiction arises from a distinction between the notions of irreducibility considered by Evans and Davies.

Note also that, according to theorem~1 in Ref.~\cite{Frigerio1977-2}, there actually is a relation between the commutant $\{\hL_\alpha, \hL_\alpha^\dag,\hH\}'$ and the uniqueness of steady states.
\begin{theorem}[Frigerio's first theorem]\label{theorem:Frigerio1}
  If a Markovian quantum system \eqref{eq:QME} has a faithful (i.e., full-rank) steady state, then $\{\hL_\alpha, \hL_\alpha^\dag,\hH\}'=\{z\id\,|\, z\in\CC\}$ implies uniqueness of the steady state.
\end{theorem}
See also Refs.~\cite{Spohn1980-52}. However, due to the important precondition that a faithful steady state exists, this criterion is usually not possible to assess for many-body systems.

\section{Davies irreducibility and Evans irreducibility}\label{sec:distinction}
Evans irreducibility \cite{Evans1977-54} and Davies irreducibility \cite{Davies1970-19} can be defined as follows \cite{Carbone2016-77,Frigerio1977-2,Evans1978-2-17}.
\begin{definition}[Evans irreducibility and Davies irreducibility]\label{def:reducibility}
  Let $\E$ be a quantum channel ($\E=e^{\L t}$ for Markovian systems) \footnote{A \emph{quantum channel} is a completely positive trace-preserving map between two operator spaces -- the most general type of discrete-time evolution for quantum systems \cite{Nielsen2000}.}. We say that a projection operator $\hP$
  \begin{align}
  	\label{eq:EvansReduce}
  	\text{Evans-reduces $\E$ if}\quad &\E^\dag(\hP)=\hP\quad\text{and}\\
  	\label{eq:DaviesReduce}
  	\text{(Davies-)reduces $\E$ if}\quad &\E(\hP\dm\hP)=\hP\E(\hP\dm\hP)\hP\ \ \forall\dm.
  \end{align}
  If the only reducing projection operators are $0$ and $\id$, we call the system Evans irreducible or (Davies) irreducible, respectively.
\end{definition}
\noindent
Alluding to (sub)harmonic functions in the potential theory of classical Markov semigroups \cite{Jacob2002}, projections that Evans or Davies-reduce $\mathcal{E}$ are also referred to as harmonic and subharmonic projections, respectively \cite{Fagnola2002-43,Umanita2005-134,Deschamps2016-28,Carbone2016-77}.
Evans irreducibility can also be characterized in terms of $\E$ instead of the adjoint channel $\E^\dag$:
\begin{proposition}[Evans-reducibility in the Schr\"{o}dinger picture]\label{prop:EvansSchroedinger}
  Evans reducibility \eqref{eq:EvansReduce} is equivalent to the existence of a projection operator $\hP\neq 0,\id$ such that $\E(\hP\dm\hP)=\hP\E(\dm)\hP$ $\forall\dm$.
\end{proposition}
\begin{proof}
With Kraus operators $\hM_k$, any quantum channel $\E$ has an operator-sum representation (a.k.a.\ Kraus decomposition) \cite{Kraus1983,Nielsen2000}
\begin{align}\label{eq:Kraus}
	\E(\dm)=\sum_k\hM_k\dm\hM_k^\dag\quad \text{with}\quad \sum_k\hM_k^\dag\hM_k=\id.
\end{align}
To prove the equality in the proposition from Evans reducibility \eqref{eq:EvansReduce}, note that if $\E^\dag(\hP)=\sum_k\hM_k^\dag\hP\hM_k=\hP$, then
\begin{align}
	&\sum_k \big[ (\hP\hM_k^\dag-\hM_k^\dag\hP) (\hM_k\hP-\hP\hM_k) \big] \nonumber \\
	=&\sum_k( \hP\hM_k^\dag\hM_k\hP+\hM_k^\dag\hP\hP\hM_k- \hM_k^\dag\hP\hM_k\hP-\hP\hM_k^\dag\hP\hM_k )
	=\hP+\hP-\hP-\hP=0.
\end{align}
As the left-hand side is a sum of positive semidefinite operators, we have $[\hP,\hM_k]=[\hP,\hM_k^\dag]=0$, and $\E(\hP\dm\hP)=\hP\E(\dm)\hP$ follows directly by commuting $\hP$ with $\hM_k$ and $\hM_k^\dag$.

To prove the other direction of the proposition, let the image of the projection operator $\hP$ be the subspace $\V\subseteq\H$ and $\V_\perp$ its orthogonal complement (the kernel of $\hP$).
Choosing $\dm=|\psi\ket\bra\psi|=\hP|\psi\ket\bra\psi|\hP$ for any state $\psi\in\V$, the equality $\E(\hP\dm\hP)=\hP\E(\dm)\hP$ gives
\begin{align}
	\sum_k\hM_k|\psi\ket\bra\psi|\hM_k^\dag
	= \sum_k\hP \hM_k|\psi\ket\bra\psi|\hM_k^\dag \hP
	\quad\Rightarrow\quad
	\sum_k |\bra\phi|\hM_k|\psi\ket|^2=0\quad \forall\phi\in\V_\perp
\end{align}
as $\bra\phi|\hP\hM_k|\psi\ket=0$.
Thus, $(\id-\hP) \hM_k \hP=0$. Now, choose $\dm=|\phi\ket\bra\phi|$ for any state $\phi\in\V_\perp$. Then,
\begin{align}
	0=\sum_k\hM_k\hP|\phi\ket\bra\phi|\hP\hM_k^\dag
	= \sum_k\hP \hM_k|\phi\ket\bra\phi|\hM_k^\dag \hP
	\quad\Rightarrow\quad
	\sum_k |\bra\psi|\hM_k|\phi\ket|^2=0\quad \forall\psi\in\V.
\end{align}
Thus, $\hP \hM_k (\id-\hP)=0$. In other words, all Kraus operators $\hM_k$ are block-diagonal and thus commute with $\hP$ such that $\E^\dag (\hP)=\sum_k\hM_k^\dag\hP\hM_k=\sum_k\hM_k^\dag\hM_k\hP=\hP$.
\end{proof}

So, an Evans-reducible system has non-trivial subspaces that evolve independently of each other. One easily sees the following direct implications of Davies and Evans (ir)reducibility.
\begin{proposition}[Implications of Davies and Evans reducibility]\label{prop:InterpretDaviesEvans}
  (a) Davies irreducibility leads to Evans irreducibility but not vice versa. Equivalently, Evans reducibility implies Davies reducibility.
  (b) Davies reducibility is equivalent to the existence of a nontrivial invariant subspace $\V\neq \{0\},\H$.
  (c) For a Markovian quantum system \eqref{eq:QME}, Evans reducibility of $\E=e^{\L t}$ is equivalent to the presence of a strong symmetry.
\end{proposition}
\begin{proof}
For the first direction of statement (a), note that if $\hP$ Evans-reduces $\E$, we have $\E(\hP\dm\hP)=\hP\E(\dm)\hP$ for all density operators $\dm$ (Proposition~\ref{prop:EvansSchroedinger}). Then, the equation will also hold for all density operators with support in the image of $\hP$, i.e., if $\dm=\hP\dm\hP$, which implies that $\hP$ Davies-reduces $\E$. For the second direction, one can easily construct counter-examples.
Statement (b) is a rephrasing of Eq.~\eqref{eq:DaviesReduce}: Let $\hP$ Davies-reduce $\E$ and $\V\subseteq\H$ be the image of $\hP$. For any $\dm$ with $\supp \dm\subseteq\V$, we have $\supp\E(\dm)\subseteq \V$, i.e., $\V$ is an invariant subspace.
For (c), recall that a strong symmetry of a Markovian system is associated with a unitary $\hU$ that commutes with the Hamiltonian $\hH$ and all Lindblad operators $\{\hL_\alpha\}$ \cite{Buca2012-14,Albert2014-89}. If $\hP$ Evans-reduces $\E=e^{\L t}$, then $\hP$ generates the strong symmetry with respect to $\hU=e^{\mri\varphi \hP}$. Conversely, for a system with a strong symmetry $\hU=e^{\mri\varphi\hat{J}}$, the generator $\hat{J}$ commutes with $\hH$ and all $\{\hL_\alpha\}$. Let the generator have the sprectral decomposition $\hat{J}=\sum_k j_k \hP_k$. Because of simultaneous diagonaliablity, it then follows that the eigenspace projections $\hP_k$ also commute with $\{\hH,\hL_\alpha\}$ and, hence, Evans-reduce the dynamics.
\end{proof}

As discussed in Refs.~\cite{Davies1970-19,Evans1978-2-17,Schrader2001-30,Umanita2005-134,Carbone2016-77} there is a close connection between Davies irreducibility and the uniqueness and faithfulness of steady states.
\begin{theorem}[Davies irreducibility and unique steady states]\label{theorem:faithful}
  (a) For $\dim\H<\infty$, a quantum channel is Davies irreducible if and only if it has a unique faithful (full-rank) steady state.
  (b) For $\dim\H=\infty$, every Davies irreducible channel either has a unique faithful steady state or it has no steady state at all.
\end{theorem}
\noindent
Statement (a) corresponds to theorem~2.3 in Ref.~\cite{Evans1978-2-17} and corollary~1 in Ref.~\cite{Umanita2005-134}. Recall that, for $\dim\H<\infty$, every quantum channel $\E$ has at least one steady state according to Brouwer's fixed point theorem. Also, the support of any steady state is an invariant subspace \cite{Baumgartner2008-41,Carbone2016-77}.
So, a system with a rank-deficient steady state cannot be Davies irreducible.
Statement (b) was shown in theorem~3.1 of Ref.~\cite{Schrader2001-30}; see also proposition~3.2 in Ref.~\cite{Carbone2016-77}.

The system in Example~\ref{example1} only satisfies Evans irreducibility but not Davies irreducibility. Thus, Theorem~\ref{theorem:faithful} does not apply, and there is no contradiction in having two orthogonal steady states. The system has a weak symmetry \cite{Baumgartner2008-41,Buca2012-14,Albert2014-89}. Specifically, we have the weak $\mathbb{Z}_2$ symmetry $\L(\hU\hat{R}\hU^\dag)=\hU\L(\hat{R})\hU^\dag$, where $\hU=\hsigma^x_1\hsigma^x_2$ is a global spin flip \cite{Barthel2020_12}.

Lastly, for $\dim\H<\infty$, the uniqueness of the steady state implies that the system is relaxing \cite{Schirmer2010-81,Wolf2012}.
\begin{theorem}[Uniqueness of the steady state and relaxation]\label{theorem:relaxing}
  For $\dim \H<\infty$, every Markovian system \eqref{eq:QME} with a unique steady state $\dm_\text{ss}$ is \emph{relaxing}, i.e., $\lim_{t\to\infty}e^{\L t}\dm=\dm_\text{ss}$ for all density operators $\dm$.
\end{theorem}
The asymptotic subspaces (a.k.a.\ limit sets) of such relaxing systems are trivial (only comprise $\dm_\text{ss}$) and the Liovillian cannot feature purely imaginary eigenvalues \cite{Schirmer2010-81,Baumgartner2008-41,Albert2016-6,Buca2022-12}.

For brevity, we will usually refer to Davies (ir)reducibility simply as \emph{(ir)reducibility} for the rest of the paper.

\section{A classical criterion for Davies irreducibility}\label{sec:classical}
In this subsection, we will gain a better understanding of (Davies) irreducibility by considering classical Markov processes induced by a quantum channel. In the classical context, there is a well-established notion of irreducibility:
\begin{definition}[Irreducibility of classical Markov processes]
  A classical Markov process with transition matrix $\P$ is called irreducible if every state can be reached from every other state, i.e., if for every pair of states $i,j$, there exists an integer $k$ such that $[\P^k]_{i,j}>0$.
\end{definition}
According to the Perron–Frobenius theorem, every irreducible Markov process for a finite state space has a unique stationary probability distribution (the Perron vector) with all elements being non-zero. We can associate each quantum channel with classical Markov processes \cite{Gorini1976-17}:
\begin{proposition}[From a quantum channel to Markov processes]
  Let $\E$ be a quantum channel and $\{|i\ket\}$ a countable orthonormal (ON) basis of the Hilbert space $\H$. Then
  \begin{align}\label{eq:Markov-P}
  	\P_{i,j}:=  \bra i| \E\big(|j\ket\bra j|\big) |i\ket
  \end{align}
  is the stochastic matrix of a classical Markov process.
\end{proposition}
\begin{proof}
Using an operator-sum representation \eqref{eq:Kraus} of $\E$, we have
\begin{subequations}
\begin{align}
	\P_{i,j}&=\sum_k \bra i|\hM_k|j\ket\bra j|\hM_k^\dag|i\ket \geq 0 \quad \forall i,j\quad\text{and} \\
	\sum_i \P_{i,j}&=\sum_{i,k} \bra j|\hM_k^\dag|i\ket\bra i|\hM_k|j\ket
	  =\sum_k\bra j|\hM_k^\dag\hM_k|j\ket=\bra j|j\ket=1\quad \forall j,
\end{align}
\end{subequations}
i.e., $\P$ is a stochastic matrix.
\end{proof}
According to the following result, irreducibility of open quantum systems can then be characterized by the irreducibility of the associated classical Markov processes: 
\begin{proposition}[From classical to quantum irreducibility]\label{prop:ReduceClassical}
  A quantum channel is irreducible if and only if, for every choice of ON basis $\{|i\ket\}$, the associated classical Markov process with transition matrix \eqref{eq:Markov-P} is irreducible.
\end{proposition}
\begin{proof}
The proposition is equivalent to the statement that there exist orthonormal states $\{|i\ket\in\H\,|\,i\in\mc{I}\}$ such that $\hP=\sum_{i\in\mc{I}} | i \ket\bra i|$ reduces the quantum channel if and only if the index set $\mc{I}$ is a closed class for the induced classical Markov process \eqref{eq:Markov-P}. By employing the operator-sum representation \eqref{eq:Kraus} of the quantum channel, both sides of the statement are equivalent to $(\id-\hP)\hM_k\hP=0$ $\forall k$.
\end{proof}

Let us revisit the quantum system in Example~\ref{example1} to demonstrate its reducibility pictorially on the basis of Proposition~\ref{prop:ReduceClassical}. We fix the ON basis $\{|\up\up\ket, |\up\down\ket, |\down\up\ket, |\down\down\ket \}$. Now, note that
\begin{subequations}\label{eq:example1-Ldiag}
\begin{align}
	\L\big(|\up\up\ket\bra\up\up|\big)&=0, \quad
	\L\big(|\up\down\ket\bra\up\down|\big)=|\up\up\ket\bra\up\up|+|\down\down\ket\bra\down\down|-2|\up\down\ket\bra\up\down|,\\
	\L\big(|\down\down\ket\bra\down\down|\big)&=0, \quad
	\L\big(|\down\up\ket\bra\down\up|\big)=|\up\up\ket\bra\up\up|+|\down\down\ket\bra\down\down|-2|\down\up\ket\bra\down\up|,
\end{align}
\end{subequations}
i.e., in the chosen basis, the Liouvillian does not generate any off-diagonal terms. We obtain the transition matrix \eqref{eq:Markov-P} of the associated classical Markov process simply by exponentiation the $4\times 4$ matrix given by Eq.~\eqref{eq:example1-Ldiag}. This yields the following diagrammatic representation of the classical Markov process.
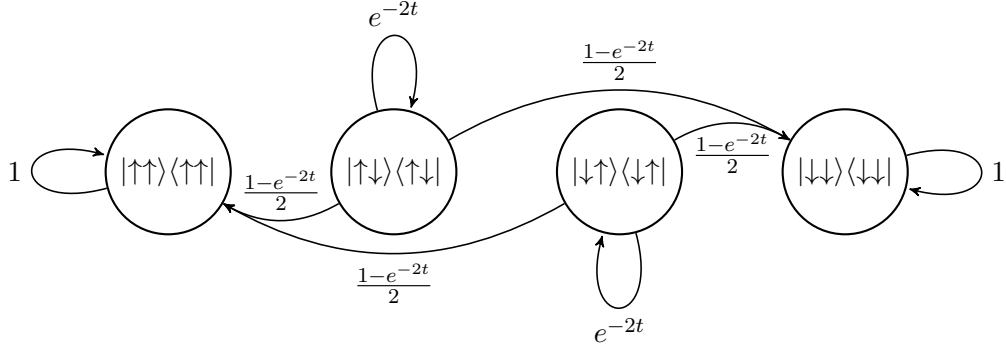
\begin{center}
	\begin{tikzpicture}[->, >=stealth', auto, semithick, node distance=3cm]
	\tikzstyle{every state}=[fill=white,draw=black,thick,text=black,scale=1]
	\node[state]  (A)                   {$|\up\up\ket\bra\up\up|$};
	\node[state]  (B)[right of=A]   {$|\up\down\ket\bra\up\down|$};
	\node[state]  (C)[right of=B]{$|\down\up\ket\bra\down\up|$};
	\node[state]  (D)[right of=C]{$|\down\down\ket\bra\down\down|$};
	\path
	(A) edge[loop left]			node{$1$}	(A)
	(B) edge[loop above]			node{$e^{-2t}$}	(B)
	(C) edge[loop below]			node{$e^{-2t}$}	(C)
	(D) edge[loop right]			node{$1$}	(D)
	(B) edge[bend left,above]	node{$\frac{1-e^{-2t}}{2}$}	(A)
	edge[bend left,above]		node{$\frac{1-e^{-2t}}{2}$}	(D)
	(C) edge[bend left,below]	node{$\frac{1-e^{-2t}}{2}$}	(A)
	edge[bend left,below]		node{$\frac{1-e^{-2t}}{2}$}	(D);
	\end{tikzpicture}
\end{center}
It is obviously reducible as both $|\up\up\ket\bra\up\up|$ and $|\down\down\ket\bra\down\down|$ are absorbing states, showing that the Markovian quantum dynamics $e^{\L t}$ is reducible.
\begin{remark}
  The irreducibility of a quantum channel requires the irreducibility of $\P$ in Eq.~\eqref{eq:Markov-P} for \emph{every} choice of the ON basis. For a reducible quantum channel, there generally still exist ON bases such that the corresponding classical Markov process is irreducible. 
\end{remark}
\begin{example}[A driven two-level system]
Consider a Markovian two-level system with $\hH=0$ and decay induced by the Lindblad operator $\hL=\hsigma^+$. The corresponding classical Markov processes \eqref{eq:Markov-P} for ON bases $\{|\up\ket,|\down\ket\}$ and $\big\{|{\rightarrow}\ket=\frac{1}{\sqrt{2}}(|\up\ket+|\down\ket),\,|{\leftarrow}\ket=\frac{1}{\sqrt{2}}(|\up\ket-|\down\ket)\big\}$ are
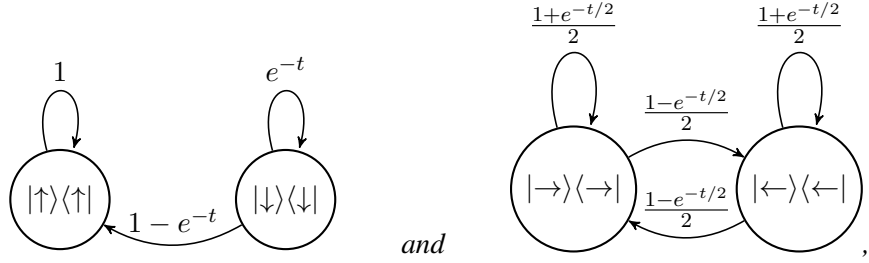
\begin{center}
\begin{tikzpicture}[->, >=stealth', auto, semithick, node distance=3cm]
	\tikzstyle{every state}=[fill=white,draw=black,thick,text=black,scale=1]
	\node[state]  (A)                   {$|\up\ket\bra\up|$};
	\node[state]  (B)[right of=A]   {$|\down\ket\bra\down|$};
	\path
	(A) edge[loop above]			node{$1$}	(A)
	(B) edge[loop above]			node{$e^{-t}$}	(B)
	(B) edge[bend left,above]	node{$1-e^{-t}$}	(A) ;
	\end{tikzpicture} \qquad and \qquad
\begin{tikzpicture}[->, >=stealth', auto, semithick, node distance=3cm]
	\tikzstyle{every state}=[fill=white,draw=black,thick,text=black,scale=1]
	\node[state]  (A)                   {$|{\rightarrow}\ket\bra{\rightarrow}|$};
	\node[state]  (B)[right of=A]   {$|{\leftarrow}\ket\bra{\leftarrow}|$};
	\path
	(A) edge[loop above]			node{$\frac{1+e^{-t/2}}{2}$}	(A)
	(B) edge[loop above]			node{$\frac{1+e^{-t/2}}{2}$}	(B)
	(A) edge[bend left,above]	node{$\frac{1-e^{-t/2}}{2}$}	(B) 
	(B) edge[bend left,above]	node{$\frac{1-e^{-t/2}}{2}$}	(A) ;
	\end{tikzpicture},
\end{center}
respectively. The first Markov process is reducible, while the second is irreducible. So, the quantum system is reducible.
\end{example}

\begin{remark}
  The irreducibility of a Markovian quantum system \eqref{eq:QME} with $\dim\H<\infty$ does not imply the matrix irreducibility of its quantum channel $\E=e^{\L t}$ or the Liouvillian $\L$ in any basis. Conversely, the Liouvillian of a Davies-reducible system is matrix-reducible. Recall that a matrix is called reducible if it is similar to a block-upper-triangular matrix by a permutation of rows and columns.
\end{remark}
\begin{example}[A two-level system with loss and gain]\label{example:loss_gain}
Consider a Markovian two-level system with $\hH=0$ and the two Lindblad operators $\hL_+=\hsigma^+$ and $\hL_-=\hsigma^-$.
For reasons explained below in Sec.~\ref{sec:algebraic}, this system is irreducible. In the ON basis $B:=\{|\up\ket\bra\down|, |\down\ket\bra\up|, |\up\ket\bra\up|, |\down\ket\bra\down|\}$ for the operator space $\A(\H)$, the Liouvillian and the associated quantum channel read
\begin{align}
	[\L]_B=
	\Bmatrix{-1&0&0&0 \\ 0&-1&0&0 \\ 0&0&-1&1 \\ 0&0&1&-1}\quad\text{and} \quad
	[e^{\L t}]_B=
	\Bmatrix{ e^{-t}&0&0&0 \\ 0&e^{-t}&0&0 \\ 0&0&\frac{1+e^{-2t}}{2}&\frac{1-e^{-2t}}{2} \\
                  0&0&\frac{1-e^{-2t}}{2}&\frac{1+e^{-2t}}{2} }.
\end{align}
Both are block upper triangular matrices, i.e., matrix reducible while the quantum system is (Davies) irreducible.
\end{example}

\begin{remark}
The conservation of a projection operator, $\L^\dag(\hP)=0$, in Evans reducibility represents a stronger restriction than Davies reducibility; c.f.\ Proposition~\ref{prop:InterpretDaviesEvans}. In fact, Evans reducibility implies that, in a suitable basis, the corresponding classical Markov process can be decoupled into two uncoupled processes as illustrated in the example below. In terms of graph theory, Evans irreducibility implies weak connectivity while Davies irreducibility requires strong connectivity for all ON bases.
\begin{example}[Dephasing for a two-level system] Consider a Markovian two-level system with $\hH=0$ and dephasing induced by the Lindblad operator $\hL=\hsigma^z$. For the ON basis $\{|\up\ket,|\down\ket\}$, the associated Markov process \eqref{eq:Markov-P} has the graphical representation
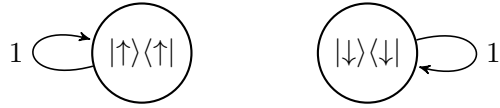
\begin{center}
\begin{tikzpicture}[->, >=stealth', auto, semithick, node distance=3cm]
	\tikzstyle{every state}=[fill=white,draw=black,thick,text=black,scale=1]
	\node[state]  (A)                   {$|\up\ket\bra\up|$};
	\node[state]  (B)[right of=A]   {$|\down\ket\bra\down|$};
	\path
	(A) edge[loop left]			node{$1$}	(A)
	(B) edge[loop right]			node{$1$}	(B);
	\end{tikzpicture}
\end{center}
which consists of two disconnected parts.
\end{example}
\end{remark}

\section{Algebraic criteria for Davies irreducibility}\label{sec:algebraic}
For Markovian quantum dynamics (a.k.a.\ quantum dynamical semigroups) generated by a Liouvillian $\L$, the (Davies) irreducibility of the channel $\E=e^{\L t}$, as captured by Definition~\ref{def:reducibility} is equivalent to the absence of non-trivial invariant subspaces for $\L$. In particular, a Markovian quantum system \eqref{eq:QME} is irreducible if and only if the only projection operators $\hP$ satisfying
\begin{align}\label{eq:Lreducible}
	\L(\hP\dm\hP)=\hP\L(\hP\dm\hP)\hP \quad  \forall \dm
\end{align}
are $0$ and $\id$. This leads us to the following algebraic reducibility criterion for the Hamiltonian and Lindblad operators.
\begin{theorem}[Davies reducibility implies invariant subspaces]\label{theorem:reducibility}
  The projection $\hP\neq 0,\id$ reduces the Markovian quantum dynamics generated by the Liouvillian \eqref{eq:QME} if and only if
  \begin{equation}\label{eq:algebraicCrit}
  	(\id-\hP)\hK\hP=0\quad\text{and}\quad
  	(\id-\hP)\hL_\alpha\hP=0 \ \forall \alpha,\quad\text{where}\quad
  	\hK:=\mri\hH+\frac{1}{2}\sum_\alpha \hL^\dag_\alpha\hL_\alpha.
  \end{equation}
\end{theorem}
\begin{proof}
Assume that $\hP$ reduces the Liouvillian. With the Lindblad form \eqref{eq:QME} of $\L$, Eq.~\eqref{eq:Lreducible} reads
\begin{align}\label{eq:Lreducible2}
	-\hK\hP\dm\hP-\hP\dm\hP\hK^\dag+\sum_{\alpha}\hL_\alpha\hP\dm\hP\hL_\alpha^\dag
	=-\hP\hK\hP\dm\hP-\hP\dm\hP\hK^\dag\hP+\sum_{\alpha}\hP\hL_\alpha\hP\dm\hP\hL_\alpha^\dag \hP.
\end{align}
Now, the `if' part of the theorem follows immediately by employing the criteria \eqref{eq:algebraicCrit} in the form $\hK\hP=\hP\hK\hP$ and $\hL_\alpha\hP=\hP\hL_\alpha\hP$ to show that both sides of Eq.~\eqref{eq:Lreducible2} agree. 
To prove the `only if' part of the theorem, let $\hP_\perp:=\id-\hP$. Setting $\dm=\hP$ in Eq.~\eqref{eq:Lreducible2}, we obtain
\begin{align}\label{eq:Lreducible3}
	-\hP_\perp\hK\hP-\hP\hK^\dag\hP_\perp+\sum_{\alpha}\big(\hL_\alpha\hP\hL_\alpha^\dag-\hP\hL_\alpha\hP\hL_\alpha^\dag\hP\big)=0.
\end{align}
Multiplying this equation from the left and right with $\hP_\perp$ gives $\sum_{\alpha}\hP_\perp\hL_\alpha\hP\hL_\alpha^\dag\hP_\perp=0$. As the left-hand side is a sum of positive semidefinite operators, it follows that $\hP_\perp\hL_\alpha\hP=0$ $\forall \alpha$. With this constraint, Eq.~\eqref{eq:Lreducible3} simplifies to $\hP_\perp\hK\hP+\hP\hK^\dag\hP_\perp=0$, which implies $\hP_\perp\hK\hP=0$ such that we have arrived at Eq.~\eqref{eq:algebraicCrit}.
\end{proof}
See also the closely related theorem~III.1 in Ref.~\cite{Fagnola2002-43} and corollary~1 in Ref.~\cite{Ticozzi2008-53}.
The condition \eqref{eq:algebraicCrit} means that all $\hL_\alpha$ and $\hK$ share a common invariant subspace. We then get an algebraic criterion for a Markovian system to be irreducible. See also Ref.~\cite{Yoshida2023_09}.
\begin{theorem}[Main result]\label{theorem:LKalgebra}
  A Markovian quantum system \eqref{eq:QME} is irreducible if the associative $\CC$-algebra generated by the Lindblad operators $\{\hL_\alpha\}$ and $\hK=\mri\hH+\frac{1}{2}\sum_\alpha \hL^\dag_\alpha\hL_\alpha$ is the whole operator space $\A(\H)$. For $\dim\H<\infty$, the converse holds as well, i.e., the ``if'' becomes an ``if and only if''.
\end{theorem}
\begin{proof}
For $\dim\H<\infty$, this is a direct consequence of Theorem~\ref{theorem:reducibility} and Burnside's theorem \cite{Burnside1905-s2,Lam1998-105,Lomonosov2004-383, Jacobson2013}. The situation is more subtle for $\dim\H=\infty$: The invariant subspaces of the $\CC$-algebra generated by $\{\hL_\alpha,\hK\}$ \cite{Note2} agree with the common invariant subspaces of $\{\hL_\alpha,\hK\}$ since linear combination and multiplication preserve subspace invariance. So, the ``if'' part follows from the fact that the full algebra $\A(\H)$ has no nontrivial subspaces. The converse direction does however not hold without further assumptions \cite{Lomonosov1991-75}. A counter-example is given by 
the algebra of all finite-rank operators on an infinite-dimensional vector space. It is distinct from $\A(\H)$ but has no nontrivial invariant subspace \cite{Shapiro2014}: For every nontrivial subspace $\V\subset\H$ containing $|\psi\ket$ and having $|\phi\ket$ in its orthogonal complement, $\V$ is not invariant under action of the rank-1 operator $|\phi\ket\bra\psi|$.
\end{proof}

The generation of the algebra in this statement involves linear combination and multiplication, but \emph{not} Hermitian conjugation \cite{Note2}.
In Sec.~\ref{sec:motivation}, we already mentioned an inaccurate interpretation of Evans theorem in the recent physics literature, based on the algebra generated by $\{\hL_\alpha,\hL_\alpha^\dag,\hH\}$ and used to assess the uniqueness of steady states. Another inaccurate approach, closer to Theorem~\ref{theorem:LKalgebra}, probes whether $\{\hL_\alpha,\hH\}$ generate the full algebra $\A(\H)$ \cite{Prosen2012-86,Prosen2012-86b,Buca2012-14,Ilievski2014-1,Oliveira2020-53,Thingna2021-31}. However, the operators $\{\hL_\alpha,\hH\}$ are generally not accessible from $\{\hL_\alpha,\hK\}$ through multiplication and linear combination, and we will give corresponding counter-examples in Sec.~\ref{sec:examples-qubit}. A case that works are systems where $\Span\{\hL_\alpha\}$ is a \emph{self-adjoint set} ($\hat{X}$ being in the set implies that $\hat{X}^\dag$ is also contained), as we can then subtract $\frac{1}{2}\sum_\alpha \hL^\dag_\alpha\hL_\alpha$ from $\hK$ to obtain $\hH$ according to Eq.~\eqref{eq:algebraicCrit}.

Employing Schur's lemma, which implies that the only operator commuting with the entire algebra $\A(\H)$ is the identity (times any scalar), we get the following corollary.
\begin{corollary}\label{coro:ncondition}
  A Markovian quantum system \eqref{eq:QME} can only be irreducible if $\{\hL_\alpha, \hK\}'=\{z\id\,|\, z\in\CC\}$.
\end{corollary}
\noindent
Note that this corollary only provides a necessary condition for irreducibility; see, e.g., Examples~\ref{example:Lind0101Lind0102}-\ref{example:LindspHsx} below. However, in the special case where $\Span\{\hL_\alpha,\hK\}$ is a self-adjoint set, the algebra generated by the Lindblad operators $\hL_\alpha$ and $\hK$ is then also a self-adjoint set. Such self-adjoint $\CC$-algebras ($C^*$-algebras) have no non-trivial invariant subspace if and only if the commutant of the algebra is just $\{z\id\}$ \cite{Arveson2001}. Generalizing this idea a bit further leads us to the following sufficient condition for irreducibility.
\begin{corollary}[Extension of Frigerio's second theorem]\label{coro:FrigerioImproved}
  Consider any subset $\mc{G}$ of the $\CC$-algebra generated by $\{\hL_\alpha,\hK\}$. If $\Span\mc{G}$ is a self-adjoint set, the Markovian quantum system \eqref{eq:QME} is irreducible if the commutant of $\mc{G}$ is trivial, i.e., if $\mc{G}'=\{z\id\,|\, z\in\CC\}$.
\end{corollary}
\noindent
A corresponding result by Frigerio is a special case of this corollary (see theorem~3.2 in Ref.~\cite{Frigerio1978-63} and theorem~5.3 in Ref.~\cite{Spohn1980-52}):
\begin{corollary}[Frigerio's second theorem]\label{coro:Frigerio2}
  If $\Span\{\hL_\alpha\}$ is a self-adjoint set and $\{\hL_\alpha\}'=\{z\id\,|\, z\in\CC\}$, then the Markovian quantum system \eqref{eq:QME} is irreducible.
\end{corollary}
\noindent
For a two-level system with loss and gain, the set of Lindblad operators $\{\hL_+=\sqrt{\gamma_+}\,\hsigma^+,\hL_-=\sqrt{\gamma_-}\,\hsigma^-\}$ meets the condition of Frigerio's second theorem as long as the rates $\gamma_\pm$ are both nonzero. Such systems, as already encountered in Example~\ref{example:loss_gain}, are then irreducible. Let us discuss further examples to illustrate these general results.

\section{Examples for Davies reducible and irreducible two-level systems}\label{sec:examples-qubit}
The following examples concern two-level systems (a qubit or spin 1/2). Throughout this section, we use the orthonormal basis $\{|\up\ket,|\down\ket\}$ to express operators, and $\{\hsigma^x,\hsigma^y,\hsigma^z\}$ denote the Pauli operators.
\begin{remark}
We will see how a change of Hamiltonian terms (Examples~\ref{example:Lind1101Hsy} and \ref{example:LindspHsx}) as well as the addition or removal of dissipators (Examples~\ref{example:Lind1101} and \ref{example:Lind1101Lind1m101}) can render a reducible system irreducible and vice versa.
Examples~\ref{example:Lind1101} and \ref{example:Lind1101Hsy} also demonstrate that a system is \emph{not} necessarily irreducible if and only if $\{\hL_\alpha,\hH\}$ generate the full operator algebra $\A(\H)$, which has been stated and used as a criterion in the recent literature (for example, in Refs.~\cite{Prosen2012-86,Prosen2012-86b,Buca2012-14,Ilievski2014-1,Oliveira2020-53,Thingna2021-31}).
In addition to Example~\ref{example1}, Examples~\ref{example:Lind0101Lind0102}-\ref{example:LindspHsx} demonstrate that a system is \emph{not} necessarily irreducible if and only if the set of operators $\{\hL_\alpha,\hL_\alpha^\dag,\hH\}$ generates the full algebra $\A(\H)$, which has been stated and used as a criterion, for example, in Refs.~\cite{Popkov2012-12,Popkov2013-15,Prosen2015-48,Manzano2018-67,Nigro2019-4}.
\end{remark}

\begin{example}\label{example:Lind1101}
  Consider the Markovian two-level system with
  \begin{equation}
  	\hH=0\ \ \text{and}\ \ \hL=\Bmatrix{1&1\\0&1}\ \ 
  	\Rightarrow\ \
  	2\hK\stackrel{\eqref{eq:algebraicCrit}}{=}\hL^\dag\hL=\Bmatrix{1&1\\1&2}.
  \end{equation}
  Clearly, the one-dimensional space spanned by $|\up\ket$ is an invariant subspace of $\hL$ and $\Span\{\hL\}$ is not self-adjoint such that the conditions for Frigerio's 2nd theorem (Corollary~\ref{coro:Frigerio2}) do not hold. However, $\hL$ and $\hK$ have no mutual invariant subspace and the $\CC$-algebra they generate \cite{Note2} is the whole operator space $\A(\H)$: For example, $\hL^2-\hL=\hsigma^+$ and $4\hK^2-2(\hK\hL+\hK)=\hsigma^-$, from which we easily obtain the operator basis $\{\id,\hsigma^x, \hsigma^y, \hsigma^z\}$. Thus, the system is irreducible according to Theorem~\ref{theorem:LKalgebra}. In accordance with Theorem~\ref{theorem:faithful}, the steady state
  \begin{equation}
  	\dm=\frac{1}{3}\Bmatrix{2&-1\\-1&1}
  \end{equation}
  is faithful, i.e., it has full rank.
  Also note that $\{\hL,\hH\}$ do \emph{not} generate $\A(\H)$.
\end{example}

\begin{example}\label{example:Lind0101Lind0102} 
  Consider the Markovian two-level system with
  \begin{equation}
  	\hH=0,\ \ \hL_1=\Bmatrix{0&1\\0&1},\ \ and\ \ \hL_2=\Bmatrix{0&1\\0&2}\ \
  	\Rightarrow\ \
  	2\hK\stackrel{\eqref{eq:algebraicCrit}}{=}\hL_1^\dag\hL_1+\hL_2^\dag\hL_2=\Bmatrix{0&0\\0&7}.
  \end{equation}
  Frigerio's 2nd theorem (Corollary~\ref{coro:Frigerio2}) does not apply since $\Span\{\hL_1,\hL_2\}$ is not self-adjoint.
  It can be checked directly that $\{\hL_1,\hL_2\}'=\{\hL_1,\hL_2,\hK\}'=\{z\id\}$ such that the condition in Corollary~\ref{coro:ncondition} is obeyed. However, $\dm=|\up\ket\bra\up|$ is the unique steady state. Since this steady state does not have full rank, the system is reducible (Theorem~\ref{theorem:faithful}). This is because Corollary~\ref{coro:ncondition} only gives a necessary condition, and, in this case, $\{\hL_1,\hL_2,\hK\}$ do not generate the full algebra $\A(\H)$. $\hP=|\up\ket\bra\up|$ reduces the quantum dynamics.
  Also note that $\{\hL_\alpha,\hL_\alpha^\dag,\hH\}$ \emph{do} generate $\A(\H)$.
\end{example}

\begin{example}\label{example:Lind1101Hsy}
  Consider the Markovian two-level system with
  \begin{equation}\label{eq:Lind1101Hsy}
  	\hH=\frac{1}{2}\hsigma^y=\frac{1}{2}\Bmatrix{0&-\mri\\\mri&0}\ \ \text{and} \ \
  	\hL=\Bmatrix{1&1\\0&1}
  	\Rightarrow\ \
  	\hK\stackrel{\eqref{eq:algebraicCrit}}{=}\mri\hH+\frac{1}{2}\hL^\dag\hL=\Bmatrix{1/2&1\\0&1}.
  \end{equation}
  Frigerio's 2nd theorem (Corollary~\ref{coro:Frigerio2}) does not apply.
  With the trivial commutant $\{\hL,\hK\}'=\{z\id\}$, the condition in Corollary~\ref{coro:ncondition} is obeyed.
  However, unlike Example~\ref{example:Lind1101}, this system is reducible according to Theorem~\ref{theorem:LKalgebra}. The inclusion of the Hamiltonian has changed $\hK$ such that $\{\hL,\hK\}$ have the common upper-triangular form \eqref{eq:Lind1101Hsy} and do not generate the full algebra $\A(\H)$ anymore. In fact, the unique steady state $\dm=|\up\ket\bra\up|$ is pure, which also implies reducibility according to Theorem~\ref{theorem:faithful}.
  Also note that $\{\hL,\hH\}$ (and hence also $\{\hL,\hL^\dag,\hH\}$) \emph{do} generate $\A(\H)$ as $[\hL,2\hH]=\mri \hsigma^z$ and $[[\hL,2\hH],2\hH]=[\mri\hsigma^z,\hsigma^y]=2\sigma^x$.
\end{example}

\begin{example}\label{example:Lind1101Lind1m101}
  Now, consider the Markovian two-level system with
  \begin{equation}
  	\hH=0,\ \ \hL_1=\Bmatrix{1&1\\0&1}, \ \ \text{and}\ \ \hL_2=\Bmatrix{1&-1\\0&1}
  	\Rightarrow\ \
  	\hK\stackrel{\eqref{eq:algebraicCrit}}{=}\frac{1}{2}(\hL_1^\dag\hL_1+\hL_2^\dag\hL_2)=\Bmatrix{1&0\\0&2}.
  \end{equation}
  Again, Frigerio's 2nd theorem (Corollary~\ref{coro:Frigerio2}) does not apply because $\Span\{\hL_1,\hL_2\}$ is not self-adjoint. In comparison to Example~\ref{example:Lind1101}, irreducibility is lost in this system through inclusion of the second dissipative process, which changed $\hK$ such that $\{\hL_1,\hL_2,\hK\}$ do not generate the full algebra $\A(\H)$.
  The unique steady state for this reducible system is again $\dm=|\up\ket\bra\up|$.
  Again, $\{\hL_\alpha,\hL_\alpha^\dag,\hH\}$ \emph{do} generate $\A(\H)$.
\end{example}

\begin{example}\label{example:LindspHsx}
  Lastly, consider the Markovian two-level system with
  \begin{equation}
  	\hH=h\hsigma^x\ \ \text{and}\ \ \hL=\hsigma^+=\Bmatrix{0&1\\0&0}
  	\Rightarrow\ \
  	2\hK\stackrel{\eqref{eq:algebraicCrit}}{=}2\mri\hH+\hL^\dag\hL=\Bmatrix{0&2\mri h\\2\mri h&1}.
  \end{equation}
  Frigerio's 2nd theorem (Corollary~\ref{coro:Frigerio2}) does not apply. For all $h\in\mathbb{R}$, including $h=0$, we have the trivial commutant $\{\hL,\hK\}'=\{z\id\}$ such that Corollary~\ref{coro:ncondition} allows for irreducibility. For $h=0$, however, the system has the unique pure steady state $\dm=|\up\ket\bra\up|$ and is reducible according to Theorems~\ref{theorem:faithful} and \ref{theorem:LKalgebra} [while, still, $\{\hL,\hL^\dag,\hH\}$ \emph{do} generate $\A(\H)$].
  For $h\neq 0$, the system is indeed irreducible and has the unique and full-rank steady state
  \begin{equation}
  	\dm=\Bmatrix{1-a&b\\b^*&a}\quad\text{with}\quad
  	a=\frac{4h^2}{1+8h^2},\ \
  	b=\frac{2\mri h}{1+8h^2}.
  \end{equation}
  With $\hL=\hsigma^+$ and $\hK\hL/h^2-2\mri K/h-2\hL=2\hsigma^-=2\hL^\dag$ we easily obtain the operator basis $\{\id,\hsigma^x, \hsigma^y, \hsigma^z\}$, showing that $\{\hL,\hK\}$ generate the full $\A(\H)$. Thus, Theorems~\ref{theorem:faithful} and \ref{theorem:LKalgebra} both establish irreducibility for $h\neq 0$. One could also employ Corollary~\ref{coro:FrigerioImproved} with $\mc{G}=\{\hL=\hsigma^+,\hK\hL/h^2-2\mri K/h-2\hL=2\hsigma^-\}$.
\end{example}

\section{Examples for Davies irreducible many-body systems I}\label{sec:examples-mb1}
Let us now exemplify how the algebraic criterion in Theorem~\ref{theorem:LKalgebra} can be used to establish the irreducibility of many-body systems. In Examples~\ref{example:tIsing} and \ref{example:XYZchain}, we consider spin-1/2 transverse Ising and XYZ chains with open boundary conditions (OBC) and boundary driving $\hL_1\propto \hsigma^+_1$ and $\hL_2\propto \hsigma^-_1$ on the first site. The algebraic proof methods used for these models are then employed in combination to establish irreducibility of a larger class of driven-dissipative spin-1/2 XYZ chains studied in the literature, where some Lindblad operators generate the full algebra on at least one or two suitable sites (Proposition~\ref{prop:XYZchains}).

\begin{example}[Boundary-driven transverse Ising chain with OBC]\label{example:tIsing}
Consider a Markovian spin-1/2 transverse Ising chain with open boundary conditions, having dissipative gain and loss terms at the left boundary,
\begin{align}
	\hH=\sum_{i=1}^Nh_i \hsigma^z_i+\sum_{i=1}^{N-1}J_i\hsigma^x_i\hsigma^x_{i+1}, \quad 
	\hL_1=\sqrt{\gamma_+}\,\hsigma_1^+,\quad \hL_2=\sqrt{\gamma_-}\,\hsigma_1^-,
\end{align}
where $\gamma_+,\gamma_->0$ and $h_i,J_i\neq 0$ $\forall i$.
We will show that the associative $\CC$-algebra $\W$ generated by $\{\hL_1,\hL_2,\hK\}$ \cite{Note2} is the  entire operator space $\A(\H)$. To begin with, we have
\begin{equation}\label{eq:tIsing-step1}
	\hsigma_1^+, \hsigma_1^- \in \W
	\quad\Rightarrow\quad
	\id,\hsigma_1^x, \hsigma_1^y, \hsigma_1^z \in \W
	\quad\Rightarrow\quad
	\hL_1^\dag,\hL_2^\dag\in \W.
\end{equation}
Hence, we can obtain $\hH$ by subtracting the Lindblad terms from $\hK$ [Eq.~\eqref{eq:algebraicCrit}], establishing that $\mri\hH=\hK-\hL_1^\dag\hL_1/2-\hL_2^\dag\hL_2/2\in \W$ and $\hH-h_1\hsigma_1^z\in \W$. We can now use commutation and multiplication to show that $\W$ contains the full algebra $\Span\{\id,\hsigma_2^x, \hsigma_2^y, \hsigma_2^z\}$ for site $i=2$:
\begin{align}
	&[\hsigma^y_1,\,\hH-h_1\hsigma^z_1 ]
	=J_1[\hsigma^y_1,\,\hsigma^x_1\hsigma^x_2]
	=-2\mri J_1 \hsigma^z_1\hsigma^x_2\in \W
	\quad\Rightarrow\quad
	\hsigma^z_1 \hsigma^z_1 \hsigma^x_2=\hsigma^x_2\in \W,\\
	&[\hsigma^x_2,\hH]=h_2[\hsigma^x_2,\hsigma^z_2]=-2\mri h \hsigma^y_2\in \W
	\quad\Rightarrow\quad
	\hsigma^x_2,\hsigma^y_2,\hsigma^z_2\in \W.
\end{align}
Subtracting all terms involving the first site from the Hamiltonian and repeating the process on site 2 establishes that $\hsigma_3^x, \hsigma_3^y, \hsigma_3^z\in \W$. Iterating the process further, we conclude that $\hsigma^x_i,\hsigma^y_i,\hsigma^z_i\in \W$ $\forall i$, implying that $\W=\A(\H)$ such that the system is irreducible according to Theorem~\ref{theorem:LKalgebra}.
\end{example}

\begin{example}[Boundary-driven spin-1/2 XYZ chain with OBC]\label{example:XYZchain}
Consider a Markovian spin-1/2 XYZ chain with open boundary conditions, having dissipative gain and loss terms at the left boundary,
\begin{align}
	\hH=\sum_{i=1}^Nh_i \hsigma^z_i+\sum_{i=1}^{N-1}(J^x_i\hsigma^x_i\hsigma^x_{i+1}+J^y_i\hsigma^y_i\hsigma^y_{i+1}+J^z_i\hsigma^z_i\hsigma^z_{i+1}),\
	\hL_1=\sqrt{\gamma_+}\,\hsigma_1^+,\ \hL_2=\sqrt{\gamma_-}\,\hsigma_1^-,
\end{align}
where $\gamma_+,\gamma_->0$ and, for every bond $(i,i+1)$, at least two of the three coupling constants $J^x_i,J^y_i,J^z_i$ are nonzero.
As in the previous example, we can show that the algebra $\W$ generated by $\{\hL_1,\hL_2,\hK\}$ is the  entire operator space $\A(\H)$. First of all, Eq.~\eqref{eq:tIsing-step1} holds and, thus, $\hH,\,\hH-h_1\hsigma_1^z\in \W$. We can now use commutation and multiplication to establish that $\W$ contains the full algebra $\Span\{\id,\hsigma_2^x, \hsigma_2^y, \hsigma_2^z\}$ for site $i=2$. In particular,
\begin{subequations}
\begin{align}
	&\hC_1^x:=[\hsigma^x_1,\,\hH-h_1\hsigma^z_1 ]/2\mri
	=J^y_1 \hsigma^z_1\hsigma^y_2 - J^z_1 \hsigma^y_1\hsigma^z_2 \in \W,\\
	&\hC_1^y:=[\hsigma^y_1,\,\hH-h_1\hsigma^z_1 ]/2\mri
	=J^z_1 \hsigma^x_1\hsigma^z_2 - J^x_1 \hsigma^z_1\hsigma^x_2 \in \W,\ \ \text{and}\\
	&\hC_1^z:=[\hsigma^z_1,\,\hH-h_1\hsigma^z_1 ]/2\mri
	=J^x_1 \hsigma^y_1\hsigma^x_2 - J^y_1 \hsigma^x_1\hsigma^y_2 \in \W.
\end{align}
\end{subequations}
If $J^x_1,J^y_1\neq 0$, we get
\begin{equation}
	\hsigma^x_2=[\hsigma_1^x,\hC_1^y]\,\hsigma_1^y/(2\mri J^x_1)\in\W,\quad
	\hsigma^y_2=[\hsigma_1^y,\hC_1^x]\,\hsigma_1^x/(2\mri J^y_1)\in\W
\end{equation}
and, hence, also $\hsigma^z_2=-\mri \hsigma^x_2\hsigma^y_2\in\W$. One can proceed similarly for the case $J^x_1,J^z_1\neq 0$ (where $J^y_1$ is allowed to be zero) and the case $J^y_1,J^z_1\neq 0$.
We can then subtract all terms involving the first site from the Hamiltonian, repeat the process to establish that $\hsigma_3^x, \hsigma_3^y, \hsigma_3^z\in \W$, and iterate to establish that $\W=\A(\H)$. So, the system is irreducible according to Theorem~\ref{theorem:LKalgebra}.
\end{example}

\begin{proposition}[Irreducibility of spin-1/2 XYZ chains]\label{prop:XYZchains}
  Markovian spin-1/2 XYZ chains with magnetic $z$ fields and open or periodic boundary conditions, Hamiltonian
  \begin{align}\label{eq:XYZchain}
	\hH=\sum_{i=1}^N (h^x_i\hsigma^x_i+h^y_i\hsigma^y_i+h^z_i\hsigma^z_i)
	   +\sum_{i}(J^x_i\hsigma^x_i\hsigma^x_{i+1}+J^y_i\hsigma^y_i\hsigma^y_{i+1}+J^z_i\hsigma^z_i\hsigma^z_{i+1}),
  \end{align}
  and Lindblad operators $\{\hL_\alpha\}$ are irreducible if,
  \begin{enumerate}[label=(\alph*)]\vspace{-0.3em}\itemsep-0.1em
   \item for every bond $(i,i+1)$, either at least two of the three coupling constants $\{J_i^x,J_i^y,J_i^z\}$ are non-zero, or exactly one coupling $J_i^{a}$ is nonzero and at least one field $h_i^{b}$ with $b\neq a$ is nonzero,
   \item the $\CC$-algebra $\W_L$ generated by $\{\hL_\alpha\}$ is a self-adjoint set, and
   \item $\W_L$ contains the full algebra $\Span\{\id,\hsigma^x_j,\hsigma^y_j,\hsigma^z_j\}$ for one boundary site ($j=1,N$) in the case of open boundary conditions, or the full algebra for any two neighboring sites $j,j+1$ (irrespective of boundary conditions).
   \end{enumerate}
\end{proposition}
\begin{proof}
In generalization of Examples~\ref{example:tIsing} and \ref{example:XYZchain}, the three conditions ensure that the algebra $\W$ generated by $\{\hL_\alpha,\hK\}$ is the full operator space $\A(\H)$ and irreducibility then follows from Theorem~\ref{theorem:LKalgebra}. 
Condition (b) ensures that $\hH\in\W$: For every Lindblad operator $\hL_\alpha$, $\W_L\subseteq\W$ contains $\hL_\alpha^\dag$ such that we can subtract $\sum_\alpha \hL_\alpha^\dag\hL_\alpha/2$ from $\hK$ [Eq.~\eqref{eq:algebraicCrit}] to show that $\hH\in\W$. Condition (c) allows us to start from the full algebra on a boundary site or two neighboring sites, to iteratively establish that $\W$ contains the full algebra for all sites $i=1,\dotsc,N$ as discussed in Examples~\ref{example:tIsing} and \ref{example:XYZchain}. According to condition (a), we will use the scheme discussed in Example~\ref{example:XYZchain} for bond $(i,i+1)$ if at least two of the three coupling constants $\{J_i^x,J_i^y,J_i^z\}$ are non-zero or, otherwise, use the scheme discussed in Example~\ref{example:tIsing}.
\end{proof}

It is remarkable that such many-body models can be rendered irreducible by having dissipators on just one or two sites.
Markovian models covered by Proposition~\ref{prop:XYZchains} have been studied with different methods, for example, in Refs.~\cite{Prosen2008-101,Benenti2009-80,Prosen2010-12,Znidaric2010-43,Znidaric2010-5,Znidaric2011-83,Prosen2011-106,Popkov2012-12,Popkov2013-2,Buca2012-14,Buca2014-112,Landi2014-90,Znidaric2014-112,Znidaric2015-92,Prosen2015-48,Schuab2016-94,Pereira2017-95,Monthus2017-4,Brenes2018-98,Reis2020-101,Oliveira2020-53,Oliveira2020-129}.
It is straightforward to translate the proposition to one-dimensional higher-spin, fermionic, and bosonic models. For a direct extension to systems in $d>1$ spatial dimensions, one needs more constraints on the set of Lindblad operators; one could, e.g., require that $\W_L$ comprises the full operator algebra for an appropriate $d-1$ dimensional hyperplane.

\section{Examples for Davies irreducible many-body systems II}\label{sec:examples-mb2}
Let us apply Theorem~\ref{theorem:LKalgebra} to a prominent class of boundary driven spin-1/2 chains with Lindblad operators like
\begin{equation}\label{eq:Lsp1sm2}
	\hL_1=\sqrt{\gamma_+}\,\hsigma_1^+\quad\text{and}\quad
	\hL_2=\sqrt{\gamma_-}\,\hsigma_N^-
\end{equation}
at the two ends $i=1,N$ such that, now, conditions (b) and (c) of Proposition~\ref{prop:XYZchains} are not obeyed. We will refer to these models as \emph{maximally boundary driven}.
Spin-1/2 XXZ chains with Lindblad operators \eqref{eq:Lsp1sm2} and further similar models are particularly interesting as they turn out to be Yang-Baxter integrable \cite{Prosen2013-15,Karevski2013-110,Ilievski2014-882,Popkov2015-114,Prosen2015-48,Ilievski2017-3,Popkov2017-95,deLeeuw2021-126}.

\begin{example}[Maximally boundary-driven spin-1/2 XX chain]\label{example:XXchainMax}
Consider a Markovian spin-1/2 XX chain with open boundary conditions, having dissipative gain and loss terms at the left and right boundaries, respectively,
\begin{align}
	\hH=\sum_{i=1}^{N}h_i\hsigma^z_i +\sum_{i=1}^{N-1}J_i\,(\hsigma^x_i\hsigma^x_{i+1}+\hsigma^y_i\hsigma^y_{i+1}),\quad
	\hL_1=\sqrt{\gamma_+}\,\hsigma_1^+,\quad \hL_2=\sqrt{\gamma_-}\,\hsigma_N^-,
\end{align}
where $\gamma_+,\gamma_->0$ and $J_i\neq 0$ $\forall i$.
We can show that the associative $\CC$-algebra $\W$ generated by $\{\hL_1,\hL_2,\hK\}$ \cite{Note2} is the  entire operator space $\A(\H)$. Starting from $\{\hsigma_1^+,\hsigma_N^-,\hK\}\subset\W$, the strategy is to commute multiple times with $\hK$ [cf.\ Eq.~\eqref{eq:algebraicCrit}] until we obtain $\hsigma_1^-$ and  $\hsigma_N^+$. One can then follow Example~\ref{example:XYZchain} to conclude that $\W=\A(\H)$.

Starting at the left boundary, we have
\begin{subequations}
\begin{align}
	&[\hsigma^+_1,\hK]-(\gamma_+/2-2\mri h_1)\hsigma^+_1 =\mri J_1 \hsigma^z_1 \hsigma^+_2\,\in\W,\\
	\Rightarrow\quad
	&[\hsigma^z_1\hsigma^+_2,\hK]-2\mri J_1\hsigma^+_1 + 2\mri h_2\hsigma^z_1 \hsigma^+_2 
	  =\mri J_2 \hsigma^z_1\hsigma^z_2 \hsigma^+_3\,\in\W,\quad \dotsc\quad ,\\
	\Rightarrow\quad
	&\Big[\prod_{j=1}^{i-1}\hsigma^z_j\,\hsigma^+_i,\hK\Big]
	  -2\mri J_{i-1}\prod_{j=1}^{i-2}\hsigma^z_j\,\hsigma^+_{i-1}
	  +2\mri h_i\prod_{j=1}^{i-1}\hsigma^z_j\, \hsigma^+_i 
	  =\mri J_i \prod_{j=1}^{i}\hsigma^z_j\, \hsigma^+_{i+1}\,\in\W
\end{align}
\end{subequations}
for $i=2,\dotsc,N-1$. So, we have $\hsigma^z_1\dotsb\hsigma^z_{N-1}\hsigma^+_N\in\W$ and, progressing similarly from the right boundary, also $\hsigma^-_1\hsigma^z_2\dotsb\hsigma^z_N\in\W$ such that
\begin{subequations}\label{eq:XXchainMax-s1m-sNp}
\begin{align}
	&(\hsigma^-_1\hsigma^z_2\dotsb\hsigma^z_N)(\hsigma^z_1\dotsb\hsigma^z_{N-1}\hsigma^+_N)=\hsigma^-_1\hsigma^+_N\in\W\\
	\Rightarrow\quad 
	&(\hsigma^-_1\hsigma^+_N)\hsigma^-_N+\hsigma^-_N(\hsigma^-_1\hsigma^+_N)=\hsigma^-_1\in\W,\quad
	(\hsigma^-_1\hsigma^+_N)\hsigma^+_1+\hsigma^+_1(\hsigma^-_1\hsigma^+_N)=\hsigma^+_N\in\W.
\end{align}
\end{subequations}
Thus, we have $\{\hsigma^+_1,\hsigma^-_1,\hsigma^+_N,\hsigma^-_N,\hK\}\in\W$ and can proceed as in Example~\ref{example:XYZchain} to conclude that $\W=\A(\H)$. Hence, the system is irreducible according to Theorem~\ref{theorem:LKalgebra}.
\end{example}

\begin{example}[Maximally boundary-driven transverse Ising chain]
Consider a Markovian spin-1/2 transverse Ising chain with open boundary conditions, having dissipative gain and loss terms at the left and right boundaries, respectively,
\begin{align}
	\hH=h\sum_{i=1}^N \hsigma^z_i+J\sum_{i=1}^{N-1}\hsigma^x_i\hsigma^x_{i+1}, \quad 
	\hL_1=\sqrt{\gamma_+}\,\hsigma_1^+,\quad \hL_2=\sqrt{\gamma_-}\,\hsigma_N^-,
\end{align}
where $\gamma_+,\gamma_->0$. For brevity, we assume uniform magnetic fields and couplings $h,J\neq 0$. Following a similar strategy as in the previous example, we will find that the algebra $\W$ generated by $\{\hL_1,\hL_2,\hK\}$ is the  entire operator space $\A(\H)$. For the following, we absorb $\gamma_+$ into $h,J$, and $\gamma_-$ ($\gamma_+\to 1$). Starting with $\{\hsigma^+_1,\hsigma^-_N,\hK\}\subset \W$, the first few steps are
\begin{subequations}
\begin{align}
	&\textstyle\hA_2:=\frac{1}{\mri J}\left([\hsigma^+_1,\hK] - (\frac{1}{2}-2h)\hsigma^+_1\right)
	= \hsigma_1^z\hsigma_2^x \in \W,\\
	\Rightarrow\quad
	&\textstyle\hB_2:=\frac{1}{2h}\,[\hA_2,\hK]
	=-\frac{J}{h}\, \hsigma_1^y + \hsigma_1^z\hsigma_2^y\in \W,\\
	\Rightarrow\quad
	&\textstyle\hA_3:=\frac{1}{2J}\,\left([\hB_2,\hK] + 2h\hA_2\right)-\frac{J}{h}\,\hsigma_1^z\hsigma_2^x
	=(1+\frac{\mri}{4h})\hsigma_1^x + \hsigma_1^z\hsigma_2^z\hsigma_3^x \in \W,\\
	\Rightarrow\quad
	&\textstyle\hB_3:=\frac{1}{2h}\,\left([\hA_3,\hK] + 2J \hB_2\right)
	=\frac{16h^2+8\mri h - 16J^2-1}{16h^2}\,\hsigma_1^y + \hsigma_1^z\hsigma_2^z\hsigma_3^y \in \W.
\end{align}
\end{subequations}
We can continue this sequence by commuting repeatedly with $\hK$, and generally have
\begin{equation}\label{eq:tIsingMax-AiBi}
	\hA_i= a_i\sigma_1^x+ \prod_{j=1}^{i-1}\hsigma^z_j\,\hsigma^x_{i}\in\W
	\quad\text{and}\quad
	\hB_i= b_i\sigma_1^y+ \prod_{j=1}^{i-1}\hsigma^z_j\,\hsigma^y_{i}\in\W
\end{equation}
for $i=2,\dotsc,N$ as, for $i\geq 3$,
\begin{subequations}
\begin{align}
	\hA_i&\textstyle\equiv \frac{1}{2J}\,\left( [\hB_{i-1},\hK] + 2h \hA_{i-1}\right) - b_{i-1} \hsigma_1^z\hsigma_2^x
	= \left(\frac{h}{J}\, a_{i-1} - \frac{\mri+4h}{4J}\,b_{i-1}\right)\sigma_1^x
	            + \prod_{j=1}^{i-1}\hsigma^z_j\,\hsigma^x_{i},\\
	\hB_i&\textstyle\equiv \frac{1}{2h}\,\left([\hA_i,\hK]+2J\hB_{i-1}\right)
	= \left(\frac{\mri +4h}{4J}\, a_{i-1} + \frac{1-8\mri h-16h^2+16J^2}{16hJ}\,b_{i-1}\right)\sigma_1^y
	            + \prod_{j=1}^{i-1}\hsigma^z_j\,\hsigma^y_{i}.
\end{align}
\end{subequations}
These two equations imply the following relation between subsequent pairs of coefficients $(a_i,b_i)$ in Eq.~\eqref{eq:tIsingMax-AiBi}.
\begin{equation}\label{eq:tIsingMax-T}
	\Pmatrix{a_i\\b_i}= T \Pmatrix{a_{i-1}\\b_{i-1}}\quad\text{with}\quad
	T = \Bmatrix{\frac{h}{J} & - \frac{\mri+4h}{4J} \\
	             \frac{\mri +4h}{4J} & \frac{1-8\mri h-16h^2+16J^2}{16hJ} }.
\end{equation}
After $N-1$ steps, we arrive at the pair of operators $\hA_{N},\hB_{N}\in\W$ and can form the linear combination
\begin{equation}\label{eq:tIsingMax-rightIter}
	b_N \hA_{N} +\mri a_N \hB_{N} - 2 a_N b_N \hsigma^+_1
	\stackrel{\eqref{eq:tIsingMax-AiBi}}{=}
	  \prod_{j=1}^{N-1}\hsigma^z_j\,\left(b_N\hsigma^x_N+\mri a_N\hsigma^y_N\right)
	=:\prod_{j=1}^{N-1}\hsigma^z_j\,\left(\alpha\hsigma^+_N+\beta\hsigma^-_N\right).
\end{equation}
Proceeding analogously, starting with $[\hsigma^-_N,\hK]$ from the right boundary, we also have 
\begin{equation}\label{eq:tIsingMax-leftIter}
	\left(\tilde{\alpha}^*\hsigma^-_1+\tilde{\beta}^*\hsigma^+_1\right)\prod_{j=2}^{N}\hsigma^z_j\in W,
	\quad\text{where}\ \
	(\tilde{\alpha},\tilde{\beta})=(\alpha,\beta)
	\ \ \text{if}\ \ \gamma_+=\gamma_-.
\end{equation}
In analogy to Eq.~\eqref{eq:XXchainMax-s1m-sNp}, the product of the operators \eqref{eq:tIsingMax-rightIter} and \eqref{eq:tIsingMax-leftIter} gives
\begin{subequations}
\begin{align}
	&\left(\tilde{\alpha}^*\hsigma^-_1+\tilde{\beta}^*\hsigma^+_1\right)\,\hsigma^z_1\hsigma^z_N
	\left(\alpha\hsigma^+_N+\beta\hsigma^-_N\right)
	= \left(\tilde{\alpha}^*\hsigma^-_1-\tilde{\beta}^*\hsigma^+_1\right)\,
	  \left(\alpha\hsigma^+_N-\beta\hsigma^-_N\right)\in\W\\
	\Rightarrow\quad
	&\left(\tilde{\alpha}^*\hsigma^-_1-\tilde{\beta}^*\hsigma^+_1\right)\,
	  \left\{\left(\alpha\hsigma^+_N-\beta\hsigma^-_N\right)\,,\,
	  \hsigma_N^-\right\}+\alpha\tilde{\beta}^*\hsigma^+_1
	 =\alpha\tilde{\alpha}^*\,\hsigma^-_1\in\W,\\
	\Rightarrow\quad
	&\left\{\left(\tilde{\alpha}^*\hsigma^-_1-\tilde{\beta}^*\hsigma^+_1\right)\,,\,
	  \hsigma_1^+\right\}\,
	  \left(\alpha\hsigma^+_N-\beta\hsigma^-_N\right)+\tilde{\alpha}^*\beta\hsigma^-_N
	 =\alpha\tilde{\alpha}^*\,\hsigma^+_N\in\W,
\end{align}
\end{subequations}
where $\{\hA,\hB\}=\hA\hB+\hB\hA$ is the anti-commutator. By analyzing the the dominant eigenvector of the transfer matrix $T$ in Eq.~\eqref{eq:tIsingMax-T}, one finds that, at least for large $N$, $\alpha\tilde{\alpha}^*$ is generally nonzero such that $\{\hsigma^+_1,\hsigma^-_1,\hsigma^+_N,\hsigma^-_N,\hK\}\in\W$, and we can proceed as in Example~\ref{example:tIsing} to conclude that $\W=\A(\H)$. Hence, the system is irreducible according to Theorem~\ref{theorem:LKalgebra}.
\end{example}

\begin{example}[Maximally boundary-driven spin-1/2 XXZ chain]
Consider a Markovian spin-1/2 XXZ chain with open boundary conditions, having dissipative gain and loss terms at the left and right boundaries, respectively,
\begin{align}
	\hH=J\sum_{i=1}^{N-1}(\hsigma^x_i\hsigma^x_{i+1}+\hsigma^y_i\hsigma^y_{i+1}+\Delta\hsigma^z_i\hsigma^z_{i+1}), \quad 
	\hL_1=\sqrt{\gamma_+}\,\hsigma_1^+,\quad \hL_2=\sqrt{\gamma_-}\,\hsigma_N^-,
\end{align}
where $\gamma_+,\gamma_->0$ and $J\neq 0$.
It has been argued that this and closely related systems have a unique steady state based on the fact that either $\{\hL_\alpha,\hH\}$ or $\{\hL_\alpha,\hL_\alpha^\dag,\hH\}$ generate the full algebra $\A(\H)$ \cite{Prosen2012-86,Prosen2013-15,Salerno2013-87,Ilievski2014-1,Prosen2015-48,Essink2020-2}. As discussed previously and illustrated in Sec.~\ref{sec:examples-qubit}, these criteria are however \emph{not} sufficient.

Nevertheless, the system has indeed a unique faithful steady state: Employing the Yang-Baxter integrability of the model \cite{Prosen2013-15}, Prosen showed that there exists a full-rank steady state (of matrix-product form) \cite{Prosen2011-106,Prosen2011-107b,Prosen2015-48} and $\{\hL_\alpha,\hL_\alpha^\dag,\hH\}'={z\id}$ as follows from the discussion in Example~\ref{example:XYZchain} or Ref.~\cite{Prosen2012-86}. Thus, the system obeys the preconditions of Frigerio's first theorem (Theorem~\ref{theorem:Frigerio1}) and is Davies irreducible.
\end{example}

\section{Dissipative stabilization of pure states}\label{sec:stabilization}
Another interesting consequence of Theorem~\ref{theorem:reducibility} concerns the stabilization of pure steady states (a.k.a.\ dark states) as studied, for example, in Refs.~\cite{Wang2001-64-6,Carvalho2001-86,Jakobczyk2002-35,Kraus2004-92,Kraus2008-78,Ticozzi2008-53,Ticozzi2009-45,Verstraete2009-5,Diehl2008-4,Diehl2010-105,Cho2011-106,Barreiro2011-470,Diehl2011-7,Ticozzi2012-370,Koga2012-85,Ikeda2013-87b,Lin2013-504,Leghtas2013-88,Reiter2013-88,Bardyn2012-109,Bardyn2013-15,Carr2013-111b,Ticozzi2014-14,Budich2015-91,Marino2016-116}. In particular, Theorem~1 of Ref.~\cite{Kraus2008-78} follows by setting $\hP=| \psi\ket\bra\psi|$:
\begin{corollary}[Stabilization of pure states]\label{coro:stabilization}
  A Markovian quantum system \eqref{eq:QME} stabilizes the pure state $|\psi\ket$ in the sense that $\L(|\psi\ket\bra\psi|)=0$ if and only if $|\psi\ket$ is a common eigenvector of all Lindblad operators $\hL_\alpha$ and $\hK=\mri\hH+\frac{1}{2}\sum_\alpha \hL^\dag_\alpha\hL_\alpha$.
\end{corollary}
\begin{proof}
For $\hP=|\psi\ket\bra\psi|$, Eq.~\eqref{eq:algebraicCrit} is equivalent to $|\psi\ket$ being an eigenstate of the $\hL_\alpha$ and $\hK$. And according to Definition~\ref{def:reducibility} and Eq.~\eqref{eq:Lreducible}, the pure-state projection $\hP=|\psi\ket\bra\psi|$ Davies-reduces the Markovian quantum dynamics if
\begin{equation}\label{eq:stabilizeProof}
	\L(\hP)=\hP\L(\hP)\hP\quad\Leftrightarrow\quad
	\L\big(|\psi\ket\bra\psi|\big)=\bra\psi|\L\big(|\psi\ket\bra\psi|\big)\,|\psi\ket\,\, |\psi\ket\bra\psi|.
\end{equation}
Taking the trace, we have $0=\Tr \L\big(|\psi\ket\bra\psi|\big)=\bra\psi|\L\big(|\psi\ket\bra\psi|\big)\,|\psi\ket$ such that Eq.~\eqref{eq:stabilizeProof} implies $\L(|\psi\ket\bra\psi|)=0$.
Hence, Corollary~\ref{coro:stabilization} follows from Theorem~\ref{theorem:reducibility}.
\end{proof}

It is an interesting question, what types of pure states can be stabilized through Markovian quantum dynamics  when imposing physical constraints on the Hamiltonian and Lindblad operators. As an example, we point out the following.
\begin{proposition}[Difficulty of stabilizing a Fermi-sea]
  Fermionic Markovian lattice systems with non-trivial local Lindblad operators (all having finite spatial support) cannot stabilize pure Fermi-sea states with particle number density $n\neq 0,1$ (no particles or completely filled bands).
\end{proposition}
\begin{proof}
Schematically, let us consider spinless quasi-free fermions on a $d$-dimensional lattice with dispersion relation $\veps(\vk)$, which is a smooth function of the single-particle wavevector $\vk$. The corresponding Fermi-sea states are of the form $|\psi\ket=\prod_{\veps(\vk)<\veps_\text{F}}\hc_{\vk}^\dag |0\ket$, where $|0\ket$ is the vacuum state, $\hc_\vk^\dag=(2\pi)^{-d/2}\sum_j e^{-\mri \vk\cdot\vec{r}_j}\hc_j^\dag$ creates a particle in quasi-momentum (Bloch) state $|\vk\ket$, and $\veps_\text{F}$ is the Fermi energy. For $|\psi\ket$ to be stabilized by the Markovian dynamics, it needs to be an eigenstate of all Lindblad operators.
By assumption, the Lindblad operators $\hL_\alpha$ are local, i.e., have a finite-range spatial support. The associated fermionic operator space has finite dimension, e.g., being spanned by $\{\id,\hc_j^\pdag,\hc_j^\dag,\hc_j^\dag\hc_j^\pdag\}$ for a single site $j$. So, when expressing $\hL_\alpha$ as a polynomial in the momentum-space ladder operators $\{\hc_\vk^\pdag,\hc_\vk^\dag\}$,
\begin{equation*}
	\hL_\alpha = \sum_{m,n} \int\ud^d k_1\dotsb\ud^d k_n\ud^d k'_1\dotsb\ud^d k'_m\,\,
	             f_{\alpha}^{(m,n)}(\vk_1,\dotsc,\vk_m,\vk'_1,\dotsc,\vk'_n)\,\,
	             \hc_{\vk_1}^\dag\dotsb\, \hc_{\vk_m}^\dag \, \hc_{\vk_n'}^\pdag\dotsb\, \hc_{\vk_1'}^\pdag,
\end{equation*}
the sum is finite and the prefactors $f_{\alpha}^{(m,n)}$ are analytic functions of the quasi-momenta. But the Fermi-sea state $|\psi\ket$ has a sharp feature at the Fermi surface, and can hence not be an eigenstate of such local Lindblad operators unless they are trivial ($\hL_\alpha\propto\id$).
\end{proof}

\section{Discussion}
\emph{Summary. -- }
We have discussed necessary and sufficient algebraic criteria for the Davies irreducibility of Markovian quantum systems \eqref{eq:QME}, which implies the absence of invariant subspaces and that such systems have unique faithful steady states (or, potentially, no steady state at all if $\dim\H=\infty$ \cite{Note1}). The criteria are based on properties of the Lindblad operators $\{\hL_\alpha\}$ and the operator $\hK=\mri\hH+\sum_\alpha \hL^\dagger_\alpha\hL_\alpha$. Davies irreducibility should not be confused with Evans irreducibility which is associated with the absence of strong symmetries but does not necessarily imply a unique steady state. As demonstrated in Sec.~\ref{sec:motivation}, one can easily construct Markovian systems that have $\{\hL_\alpha, \hL_\alpha^\dag,\hH\}'=\{z\id\}$ and are, hence, Evans irreducible but have multiple steady states.
In various examples we saw that a change of Hamiltonian terms as well as the addition or removal of dissipators can render a Davies reducible system irreducible and vice versa. Also, one-dimensional many-body models can be rendered irreducible by having suitable dissipators on just one or two sites. Furthermore, if $\{\hL_\alpha,\hL_\alpha^\dag,\hH\}$ or even $\{\hL_\alpha,\hH\}$ generate the full algebra $\A(\H)$, the system may still be Davies reducible. So, these latter criteria alone, as employed in recent works, are insufficient to establish uniqueness of steady states. However, the Davies irreducibility of maximally boundary-driven spin-1/2 XXZ chains \cite{Prosen2012-86,Prosen2015-48} follows also from the existence of a full-rank (matrix-product) steady state \cite{Prosen2011-106,Prosen2011-107b,Prosen2015-48} and Frigerio's first theorem (Theorem~\ref{theorem:Frigerio1}).

\emph{Historical note. --}
The confusion concerning Evans irreducibility and the uniqueness of steady states, as described in Sec.~\ref{sec:motivation}, may in fact originate from Evans' paper \cite{Evans1977-54}. Under \S3, it says
\emph{``We note that the argument of [3, Theorem 13], can be used to show the following. If $T_t$ is a dynamical semigroup on a $W^*$-algebra $\mc{A}$ which leaves invariant no proper hereditary [13]\, $^*$-subalgebra of $\mc{A}$, then $T_t$ has no non-scalar fixed points i.e.\ $\{x\in\mc{A}:T_t(x)=x,\, \forall t\}=\CC$''}
but associates invariance/reduction with Evans reducibility \eqref{eq:EvansReduce} as can also be seen in the comment
\emph{``We say that the hereditary $W^*$-subalgebra $\mc{B}$ reduces $T_t$
if the semigroup $T_t$ leaves $\mc{B}$ globally invariant and its restriction to $\mc{B}$ is a dynamical semigroup.
Note that the hereditary $W^*$-subalgebra $p\mc{A} p$ reduces $T_t$ if and only if $T_t p= p$ for all positive $t$.''}
However, \emph{``[3, Theorem 13]''} refers to Davies' paper \cite{Davies1970-19} and requires Davies irreducibility, i.e., the absence of nontrivial invariant subspaces (cf.\ Proposition~\ref{prop:InterpretDaviesEvans}).

In Evans' discussion, $T_t$ denotes the adjoint channel $\E^\dag=e^{\L^\dag t}$, and $T_t$ having only scalar fixed points $z\id$ implies uniqueness of the steady state (if one exists). 

\emph{Note added. --} A few days before the submission of this paper, we were made aware of the preprint~\cite{Yoshida2023_09}. It gives an alternative proof of Theorem~\ref{theorem:LKalgebra} for $\dim\H<\infty$ and discusses implications for systems with strong symmetries like spin systems with dephasing.

\vspace{1em}
\begin{acknowledgments}\vspace{-1.2em}
We gratefully acknowledge valuable feedback on the manuscript by Federico Carollo, Igor Lesanovsky, and Lorenza Viola.
\end{acknowledgments}

\end{document}